\documentclass{amsart}

\usepackage{amsmath,amssymb,amsbsy,amsthm}
\usepackage{graphicx,float}

\theoremstyle{plain}
\newtheorem{theorem}{Theorem}[section]

\newtheorem{corollary}[theorem]{Corollary}

\theoremstyle{definition}
\newtheorem{definition}{Definition}[section]
\newtheorem{example}{Example}[section]

\theoremstyle{remark}
\newtheorem*{remark}{Remark}

\numberwithin{equation}{section}

\begin{document}

\title[WEIGHTED AUTOMATA \& RECURRENCE EQUATIONS FOR REGULAR LANGUAGES]{WEIGHTED AUTOMATA AND RECURRENCE EQUATIONS FOR REGULAR LANGUAGES}
\author[E. Carta-Gerardino \& P. Babaali]{E. Carta-Gerardino$^1$ \\ ecarta-gerardino@york.cuny.edu \\ \\ P. Babaali$^2$ \\ pbabaali@york.cuny.edu \\ \\ $^{1,2}$Department of Mathematics and Computer Science \\ York College, City University of New York \\ 94-20 Guy R. Brewer Boulevard, Jamaica, New York 11451 \\ United States}

\begin{abstract}Let $\mathcal{P}(\Sigma^*)$ be the semiring of languages, and consider its subset $\mathcal{P}(\Sigma)$. In this paper we define the language recognized by a weighted automaton over $\mathcal{P}(\Sigma)$ and a one-letter alphabet. Similarly, we introduce the notion of language recognition by linear recurrence equations with coefficients in $\mathcal{P}(\Sigma)$. As we will see, these two definitions coincide. We prove that the languages recognized by linear recurrence equations with coefficients in $\mathcal{P}(\Sigma)$ are precisely the regular languages, thus providing an alternative way to present these languages. A remarkable consequence of this kind of recognition is that it induces a partition of the language into its cross-sections, where the $n$th cross-section contains all the words of length $n$ in the language. Finally, we show how to use linear recurrence equations to calculate the density function of a regular language, which assigns to every $n$ the number of words of length $n$ in the language. We also show how to count the number of successful paths of a weighted automaton. \end{abstract}


\maketitle

\section{Introduction}

Weighted automata are powerful finite-state machines in which every transition carries a weight from a semiring. These automata have been studied recently in a wide range of settings, from very applied fields, like natural language and speech-processing (see \cite{Per96,Moh00,Moh04}), to more theoretical ones, like logic (see \cite{Dro05}). In our current research, we are particularly interested in the applications of weighted automata to formal language theory.

A finite automaton (\cite{Eil74,Kho01}) can be regarded as a particular type of weighted automaton, by letting the weights come from the \textit{Boolean semiring} (i.e., the weights are either 0 or 1). Thus, the class of weighted automata contains the class of finite automata. Kleene's Theorem states that finite automata recognize the regular languages. Hence, it is no surprise that weighted automata can be used to recognize a class of languages that contains the class of the regular languages. In particular, it can be shown that weighted automata can be used to recognize context-free languages (see \cite{Cor99}).

In our work we are interested in weighted automata over a one-letter alphabet. We refer to these automata as counting automata, since they can be used as counting devices, with applications in combinatorics and enumeration (see \cite{Rut01,Rue02,Rut02}), among others. We start by recalling the definitions of a semiring and a formal power series (Section 2). These notions provide the setting we need to associate a linear recurrence equation to each state of a counting automaton. In fact, we will see that a counting automaton over a semiring $K$ generates a system of linear recurrence equations with coefficients in $K$ (Section 3).

Given our interest in the applications of weighted automata to formal language theory, we explore counting automata, and recurrence equations, over the \textit{semiring of languages}, $\mathcal{P}(\Sigma^*)$ (Section 4). Specifically, we consider its subset $\mathcal{P}(\Sigma)$. We define the language recognized by a counting automaton over $\mathcal{P}(\Sigma)$, and introduce the idea of language recognition by linear recurrence equations with coefficients in $\mathcal{P}(\Sigma)$. We will see that these two types of language recognition are equivalent. A consequence of recognizing a language this way is that we obtain a partition of the language into its \textit{cross-sections}, where the $n$th cross-section contains all the words of length $n$ in the language (see \cite{Ack07,Ack09}). It is important to notice that this is the case because the weights of the automata and the coefficients of the recurrence equations come from $\mathcal{P}(\Sigma)$. We will show that the languages recognized by counting automata over $\mathcal{P}(\Sigma)$, and by linear recurrence equations with coefficients in $\mathcal{P}(\Sigma)$, are closed under certain operations. We then prove that a language recognized by a system of linear recurrence equations with coefficients in $\mathcal{P}(\Sigma)$ is regular, and that every regular language is recognized by a system of linear recurrence equations with coefficients in $\mathcal{P}(\Sigma)$. This result provides a novel way to present this important class of languages.

We conclude this paper by showing how to use linear recurrence equations to count, for every $n$, the number of words of length $n$ in a regular language (Section 5). That is, we show how to calculate the \textit{density function} of the language (see \cite{Roz97}). We will see that the number of words of length $n$ in a language is closely related to the number of successful paths of length $n$ in the counting automaton recognizing the language. Thus, we start by counting the number of successful paths of any given length in an automaton. We do this by constructing an automaton that counts the number of successful paths of another automaton. We refer to this machine as a path-counting automaton, and we use it to construct the self-counting automaton, which we will define as a machine with the ability to count its own successful paths. Therefore, for every $n$, we have a way to (\textit{i}) generate all the words of length $n$ and to (\textit{ii}) count the number of words of length $n$ in a regular language.

\section{Preliminaries: Semirings and Formal Power Series}

A \textbf{monoid} is a nonempty set on which we define an associative operation, and in which there is an identity element. Using this, we can define a semiring.

\begin{definition}
A \textbf{semiring} $K = (K, +, \cdot, 0, 1)$ is a set $K$ satisfying
\begin{enumerate}
\item $(K, +, 0)$ is a commutative monoid with identity element 0
\item $(K, \cdot, 1)$ is a monoid with identity element 1
\item $\text{for all } a,b,c \text{ in } K$, $a \cdot (b+c) = a \cdot b + a \cdot c$
\item $\text{for all } a \text{ in } K$, $0 \cdot a = a \cdot 0 = 0$
\end{enumerate}
\end{definition}

From the definition we can see that every ring with unity is a semiring. (For example, the ring of the real numbers is an example of a semiring.) Some nontrivial examples of semirings are the \textbf{semiring of natural numbers} $\mathbb{N} = (\mathbb{N}, +, \cdot, 0, 1)$, the \textbf{Boolean semiring} $\mathbb{B} = (\{0,1\}, \vee, \wedge ,0 ,1)$, and the \textbf{semiring of languages} $\mathcal{P}(\Sigma^*) = (\mathcal{P}(\Sigma^*), \cup, \cdot, \emptyset, \{\varepsilon\})$, where $\Sigma$ is a finite \textit{alphabet}, $\Sigma^*$ is the set of all \textit{words} of finite length over $\Sigma$ ($\varepsilon$ denotes the empty word), and $\mathcal{P}(\Sigma^*)$ is the power set of $\Sigma^*$, known as the set of \textit{languages} over $\Sigma$. It is not difficult to see that if $K_1$ and $K_2$ are two semirings, then their direct product $K_1 \times K_2$ is also a semiring.

\begin{definition}
Let $A$ be a finite alphabet and $K$ a semiring. We can define a map $s: A^* \rightarrow K$, assigning to every word $w \in A^*$ an element $c \in K$. Such a map is known as a \textbf{formal power series}.
\end{definition}

We call $s(w) = c$ the \textit{coefficient} of $w$, or the \textit{weight} of $w$. Of course, these coefficients or weights have different interpretations, depending on the particular semiring $K$. The set of all formal power series $s: A^* \rightarrow K$ is usually denoted by $K\left\langle \left\langle A^* \right\rangle \right\rangle$.

For example, if $A$ is a finite alphabet and $K = \mathbb{B}$, then notice that for $w \in A^*$, $s(w)$ is either 0 or 1 (false or true, respectively). Hence, a formal power series $s \in \mathbb{B}\left\langle \left\langle A^* \right\rangle \right\rangle$ rejects or accepts a word $w \in A^*$.


We have seen that, given a semiring $K$ (and a finite alphabet $A$), we can define the set of formal power series $K\left\langle \left\langle A^* \right\rangle \right\rangle$. In turn, the set of formal power series can be made into a semiring in the following way (\cite{Kui86}). Addition of two series $s_1, s_2 \in K\left\langle \left\langle A^* \right\rangle \right\rangle$ is defined by $(s_1 + s_2)(w) = s_1(w) + s_2(w)$, for all $w \in A^*$. The series defined by $\mathbf{0}(w) = 0$ is the identity for the addition. Multiplication of two series $s_1, s_2$ is defined by $(s_1 \cdot s_2)(w) = \displaystyle \sum_{w_1w_2=w} (s_1(w_1))\cdot(s_2(w_2))$, for all $w \in A^*$. (This operation is known as the Cauchy product of two formal power series.) The identity for the product is the series $\pmb{\varepsilon}$ defined by $\pmb{\varepsilon}(\varepsilon) = 1$, while $\pmb{\varepsilon}(w) = 0$ for any other word $w \neq \varepsilon$. Hence, $(K\left\langle \left\langle A^* \right\rangle \right\rangle, +, \cdot, \mathbf{0}, \pmb{\varepsilon})$ is a semiring, the \textbf{semiring of formal power series}.

\section{Weighted Automata and Recurrence Equations}

A convenient way to represent some formal power series is by means of weighted automata (\cite{Dro05}).

\begin{definition}
Let $K = (K, +, \cdot, 0, 1)$ be a semiring and $A$ a finite alphabet. A \textbf{weighted automaton} $\mathcal{A}$ over $K$ and $A$ is a quadruple $\mathcal{A} = (Q_{\mathcal{A}}, \iota, \tau, \varphi)$, where $Q_{\mathcal{A}}$ is a finite set of states, $\iota, \varphi: Q_{\mathcal{A}} \rightarrow K$ are functions defining the initial weight and the final weight of a state, respectively, and if $n$ is the number of states, $\tau: A \rightarrow K^{n \times n}$ is the transition weight function. We let $\tau(x)$ be an $(n \times n)$-matrix whose $(i,j)$-entry $\tau(x)_{i,j} \in K$ gives the weight of the transition $q_i \stackrel{x}{\longrightarrow} q_j$. If $\tau(x)_{i,j} = a$, we denote this by $q_i \stackrel{x|a}{\longrightarrow} q_j$.
\end{definition}

Notice that the definition of a weighted automaton does not include the notions of initial or final states. However, by appropriately defining $\iota$ and $\varphi$, it is possible to equip a weighted automaton with initial and final states, as we will see later on.

Consider now the path $P: q_0 \stackrel{x_1|a_1}{\longrightarrow} q_1 \stackrel{x_2|a_2}{\longrightarrow} q_2 \longrightarrow \ldots \longrightarrow q_{n-1} \stackrel{x_n|a_n}{\longrightarrow} q_n$ in $\mathcal{A}$. Denote the \textbf{length} of the path by $|P|$. Now define the \textbf{weight} of the path by
\begin{displaymath} \left\| P \right\| = \iota(q_0) \cdot a_1 \cdot  a_2 \cdots a_n \cdot \varphi(q_n). \end{displaymath}
Notice that this path has as label the word $w = x_1 x_2 \ldots x_n \in A^*$. There might be, of course, other paths with label $w = x_1 x_2 \ldots x_n$. We will define the weight of the word $w$ in $\mathcal{A}$ to be the sum of the weights $\left\| P \right\|$ over all paths $P$ with label $w$. Denote this by $\left\| \mathcal{A} \right\| (w)$, and notice that the weight of a word $w$ in $\mathcal{A}$ is a function from $A^*$ to $K$. That is, $\left\| \mathcal{A} \right\|$ is a formal power series, so $\left\| \mathcal{A} \right\| \in K\left\langle \left\langle A^* \right\rangle \right\rangle$.

A formal power series $s \in K\left\langle \left\langle A^* \right\rangle \right\rangle$ is said to be \textit{automata recognizable} if there is a weighted automaton $\mathcal{A}$ such that $s = \left\| \mathcal{A} \right\|$. In this case we say that $\mathcal{A}$ is an \textit{automata representation} for $s \in K\left\langle \left\langle A^* \right\rangle \right\rangle$.

In our research we are interested in automata over a one-letter alphabet $A = \{x\}$. Hence, a typical path in such an automaton is $P: q_0 \stackrel{x|a_1}{\longrightarrow} q_1 \stackrel{x|a_2}{\longrightarrow} q_2 \longrightarrow \ldots \longrightarrow q_{n-1} \stackrel{x|a_n}{\longrightarrow} q_n$. Given that every transition reads the letter $x$, we eliminate it from the diagram for simplicity, thus making a typical path look like $P: q_0 \stackrel{a_1}{\longrightarrow} q_1 \stackrel{a_2}{\longrightarrow} q_2 \longrightarrow \ldots \longrightarrow q_{n-1} \stackrel{a_n}{\longrightarrow} q_n$. Since $A = \{x\}$, an arbitrary word $w \in A^*$ has the form $w = x^n$ for some $n \in \mathbb{N}$. By definition, $\left\| \mathcal{A} \right\|(w) = \left\| \mathcal{A} \right\|(x^n)$ equals the sum of $\left\| P \right\|$ over all paths $P$ with label $w = x^n$. But since every transition reads the letter $x$, $\left\| \mathcal{A} \right\|(x^n)$ equals the sum of $\left\| P \right\|$ over all paths $P$ of length $n$. We call $\left\| \mathcal{A} \right\|$ the \textbf{behavior of the automaton} $\mathcal{A}$, and define it as
\begin{displaymath} \left\| \mathcal{A} \right\|(x^n) = \sum_P \{ \left\| P \right\| : |P|=n \}. \end{displaymath}
Note that in this kind of automaton we are not directly accepting/rejecting words over some alphabet, but rather counting paths of length $n$, and keeping track of the weights of such paths. The idea of using automata as counting devices has been used recently with applications in combinatorics and difference equations \cite{Rut01,Rue02,Rut02}. Thus, we refer to weighted automata over a one-letter alphabet as \textbf{counting automata}. In our work, we further explore some of the properties of counting automata.

Suppose we are interested in computing the weight of all paths of length $n$ (equivalently, all paths with label $x^n$) in an automaton $\mathcal{A}$, starting at a specific state $q_0$. Then we would look at all paths of length $n$ starting at $q_0$, compute the weight of each of these paths, and add up these weights. We call this the \textbf{behavior of the state} $q_0$ and denote it by $\left\| \mathcal{A} \right\|_{q_0}$. Then
\begin{displaymath} \left\| \mathcal{A} \right\|_{q_0}(x^n) = \sum_P \{ \left\| P \right\| : |P|=n \text{ and } P \text{ starts at } q_0 \}. \end{displaymath}
Notice that, given any state $q_0$, $\left\| \mathcal{A} \right\|_{q_0} \in K\left\langle \left\langle \{x\}^* \right\rangle \right\rangle \cong K^{\mathbb{N}}$. Since $\left\| \mathcal{A} \right\|_{q_0}$ assigns to every word $x^n$ an element $c_n \in K$, we can identify $\left\| \mathcal{A} \right\|_{q_0}: x^n \mapsto c_n$ with a function $f_0: n \mapsto c_n$. Hence, in a counting automaton, every state $q_0$ generates a function $f_0 \in K^{\mathbb{N}}$, and thus we can identify each state with the function it generates. The idea of associating a function to each state of an automaton goes back to classical automata theory (see \cite{Brz64,Eil74}).

In what follows, we will assume that the initial weights of the states of an automaton $\mathcal{A}$ are either 0 or 1. Those with a weight of 1 will be the initial states, and those with a weight of 0 will be non-initial. Denote the set of initial states by $I_{\mathcal{A}}$. For the moment, suppose that $I_{\mathcal{A}} = Q_{\mathcal{A}}$, so that every state is allowed to be an initial state. In the next section we will assume that $I_{\mathcal{A}} \subsetneq Q_{\mathcal{A}}$. We will allow more freedom to the way we define the final weights of the states of an automaton. Those states with a non-zero weight will be the final states, and those with a weight of 0 will be non-final. We will denote the set of final states by $F_{\mathcal{A}}$.

Now consider an arbitrary state $f_0 \in Q_{\mathcal{A}}$. Let $\{f_1,\ldots,f_k\}$ be the set of states that can be reached from $f_0$ through paths of length 1, with transition weights $a_1,\dots,a_k$, respectively. Suppose that the final weight of state $f_0$ is $c_0$, and that the final weights of states $f_1,\ldots,f_k$ are $c_1,\ldots,c_k$, respectively. A graphical representation of this is

\begin{figure}[H]
\centering
\includegraphics[scale=1]{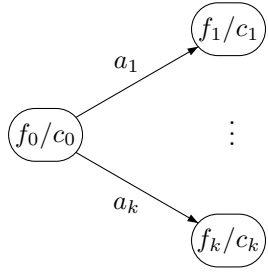}
\caption{Paths of length 1 starting at state $f_0$}
\label{fig:AUTlengthoneCON.eps}
\end{figure}

Let $f_0 \in K^{\mathbb{N}}$ be the function generated by state $f_0$, and let $f_1,\ldots,f_k \in K^{\mathbb{N}}$ be the functions generated by states $f_1,\ldots,f_k$, respectively. It can be shown (see \cite{Rut01}), that
\begin{equation} \label{eq:EqnsRec}
\begin{aligned}
&f_0(0)= c_0, \\
&f_0(n+1)= a_1 f_1(n) + a_2 f_2(n) + \ldots + a_k f_k(n), \text{ for } n \geq 0.
\end{aligned}
\end{equation}

Notice that Eqs. \ref{eq:EqnsRec} provide a recursive definition of the function generated by each state of a counting automaton. Using this, it can be shown that a counting automaton $\mathcal{A}$ over $K$ generates a system of linear recurrence equations. And by definition, these are the only equations recognized by $\mathcal{A}$. 


\begin{theorem} \label{TheoremAutRec} (\cite{Rut01}) Suppose that $\mathcal{A}$ is a counting automaton over a semiring $K$.
\begin{figure}[H]
\centering
\includegraphics[scale=1]{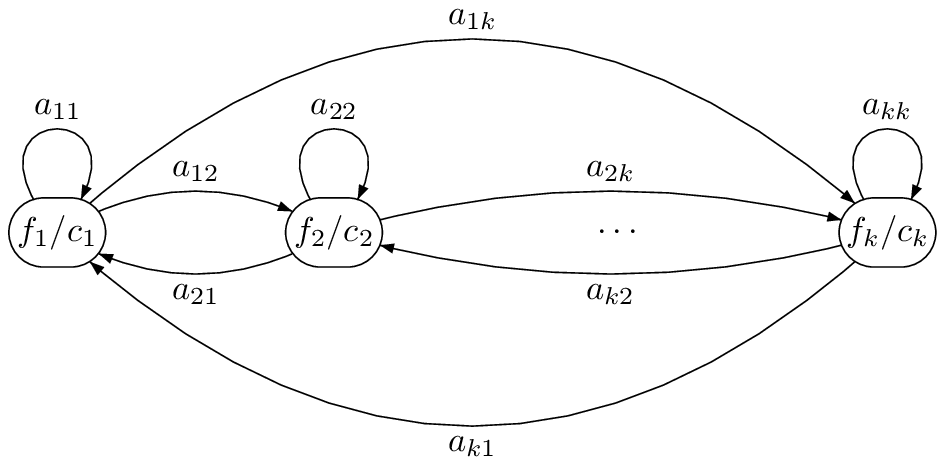}
\label{fig:Weighted1.eps}
\end{figure}
\noindent The functions $f_1,f_2,\ldots,f_k \in K^{\mathbb{N}}$ generated by $\mathcal{A}$ satisfy the following system of linear recurrence equations.
\begin{equation} \label{eq:SysAutRec} f_i(n+1) = \displaystyle \sum_{j=1}^k a_{ij} f_j(n), \quad f_i(0) = c_i, \quad 1 \leq i \leq k \end{equation}
\noindent Conversely, given this system of linear recurrence equations, the counting automaton recognizing it is precisely $\mathcal{A}$.
\end{theorem}

\begin{example} \label{ExHigherDegree} Higher-Degree Systems

Consider the following system of linear recurrence equations over an arbitrary semiring $K$.
\begin{equation*} \label{eq:SysEx1}
\begin{tabular}{l r l}
$f_1(n+4) =$ & $a_{12}f_2(n),$ & $f_1(0) = c_1$ \\
$f_2(n+1) =$ & \hspace{-0.35cm} $a_{21}f_1(n) + a_{22}f_2(n),$ & $f_2(0) = c_2$
\end{tabular}
\end{equation*}

Note that the degree of $f_1$ is 4. Theorem \ref{TheoremAutRec} guarantees that we can build an automaton recognizing equations of degree 1. In order to use this result, we need to introduce additional functions that act as intermediate states. The functions we need can be defined as follows.
\begin{equation*} \label{eq:SysEx1a}
\begin{tabular}{l l}
$g_1(n+1) = f_1(n+2),$ & $g_1(0) = f_1(1) = d_1$ \\
$g_2(n+1) = f_1(n+3),$ & $g_2(0) = f_1(2) = d_2$ \\
$g_3(n+1) = f_1(n+4),$ & $g_3(0) = f_1(3) = d_3$
\end{tabular}
\end{equation*}

Using these auxiliary functions, we can rewrite the original system of equations as a system of equations of degree 1.
\begin{equation*} \label{eq:SysEx1b}
\begin{tabular}{l l l l l l}
$f_1(n+1) =$ & & $g_1(n)$ \\
$g_1(n+1) =$ & & & $g_2(n)$ \\
$g_2(n+1) =$ & & & & $g_3(n)$ \\
$g_3(n+1) =$ & & & & & $a_{12}f_2(n)$ \\
$f_2(n+1) =$ & \hspace{-0.35cm} $a_{21}f_1(n) +$ & & & & $a_{22}f_2(n)$
\end{tabular}
\end{equation*}

Now we can use Theorem \ref{TheoremAutRec} to build the automaton that recognizes the given system of linear recurrence equations.

\begin{figure}[H]
\centering
\includegraphics{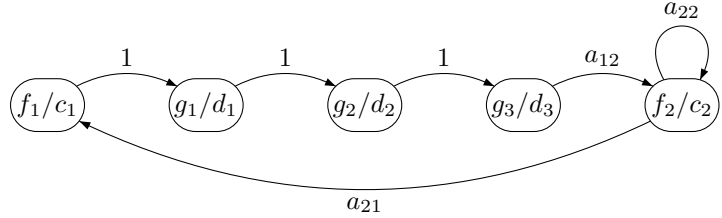}
\caption{Automaton recognizing the system of recurrence equations in Example \ref{ExHigherDegree}}
\label{fig:EXdeg4.eps}
\end{figure}

\end{example}

Theorem \ref{TheoremAutRec} above shows that a counting automaton over a semiring $K$ generates a system of linear recurrence equations with coefficients in $K$. In the next section we will restrict our attention to the case where $K = \mathcal{P}(\Sigma^*)$. Specifically, we will consider its subset $\mathcal{P}(\Sigma)$. Our goal is to define the language recognized by a counting automaton over $\mathcal{P}(\Sigma)$, and to define what it means for a language to be recognized by a system of linear recurrence equations with coefficients in $\mathcal{P}(\Sigma)$. One of the implications of defining languages this way is that we obtain an immediate partition of the language into its cross-sections. We will see that it is also possible to define these languages through formal grammars, and we will show that these languages are closed under union, concatenation, and the Kleene star. Using this, we will prove that the languages recognized by linear recurrence equations with coefficients in $\mathcal{P}(\Sigma)$ are precisely the regular languages.

\section{Language Recognition, Language Partition, and the Cross-Sections of a Regular Language}

In this section, the semiring we use for the weights of the automata and for the coefficients of the recurrence equations is $\mathcal{P}(\Sigma^*)$. In particular, we will only consider weights and coefficients in $\mathcal{P}(\Sigma) \subset \mathcal{P}(\Sigma^*)$.

Suppose that $\mathcal{A}$ is a counting automaton with weights in $\mathcal{P}(\Sigma)$, and assume that the set of states of $\mathcal{A}$ is $\{f_1, f_2, \ldots, f_k \}$. As it is customary when using automata for language recognition, we will assume that there is only one initial state. Without loss of generality, we will let the first state be the initial state. Therefore, $I_{\mathcal{A}} = \{ f_1 \}$. We now specify the final weights of the states in $\mathcal{A}$. Non-final states were defined as states that have a final weight of 0; in the semiring of languages, $\emptyset$. Final states were defined as states with a non-zero final weight. Specifically, we will assume that the final states have a final weight of 1; in the semiring of languages, $\{ \varepsilon \}$. Therefore, $f_i(0) = \emptyset$ if $f_i \notin F_{\mathcal{A}}$, and $f_i(0) = \{\varepsilon\}$ if $f_i \in F_{\mathcal{A}}$.


Let $\mathcal{A}$ be a counting automaton with $k$ states
\begin{figure}[H]
\centering
\includegraphics[scale=1]{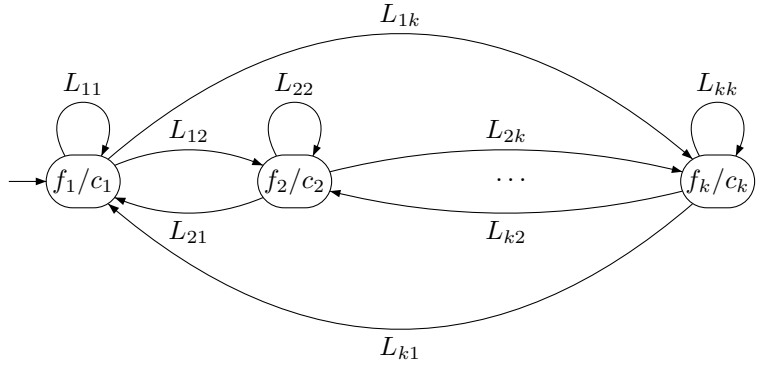}
\caption{Counting automaton $\mathcal{A}$ over $\mathcal{P}(\Sigma) \subset \mathcal{P}(\Sigma^*)$}
\label{fig:WeightedLanguages1.eps}
\end{figure}
\noindent where $I_{\mathcal{A}} = \{ f_1 \}$, $c_i = \{\varepsilon\}$ if $f_i \in F_{\mathcal{A}}$, $c_i = \emptyset$ if $f_i \notin F_{\mathcal{A}}$, and for every $i$ and every $j$, $L_{ij} \in \mathcal{P}(\Sigma)$. (If $L_{ij} = \emptyset$, we can eliminate this transition from the diagram.) We know that $f_1$, the initial state, generates a function $f_1 : \mathbb{N} \to \mathcal{P}(\Sigma^*)$ defined by the system
\begin{equation} \label{eq:SysLanguages} f_i(n+1) = \displaystyle \bigcup_{j=1}^k L_{ij} \cdot f_j(n), \quad f_i(0) = c_i, \quad 1 \leq i \leq k. \end{equation}
\noindent Since $L_{ij} \in \mathcal{P}(\Sigma)$, we have that for every $n \in \mathbb{N}$, $f_1(n)$ is a language containing words of length $n$. Denote $f_1(n)$ by $\mathcal{L}_n$.

\begin{definition}
The \textbf{language recognized by a counting automaton $\mathcal{A}$ over $\mathcal{P}(\Sigma)$} is denoted by $\mathcal{L}_{\mathcal{A}}$ and is defined by $\mathcal{L}_{\mathcal{A}} = \displaystyle \bigcup_n \mathcal{L}_n$.
\end{definition}

Thus, a word $w$ of length $n$ belongs to $\mathcal{L}_{\mathcal{A}}$ if $w$ belongs to $\mathcal{L}_n$. We will say that a word $w$ of length $n$ is recognized by the counting automaton $\mathcal{A}$ if there is a path of length $n$ starting at $f_1$ and ending at a final state with weight $\{w\}$.

Given that the languages $\mathcal{L}_n$ are defined recursively, we can also define the language $\displaystyle \bigcup_n \mathcal{L}_n$ in the following way.

\begin{definition}
$\mathcal{L} = \displaystyle \bigcup_n \mathcal{L}_n$ is known as the \textbf{language recognized by linear recurrence equations}, since the languages $\mathcal{L}_n$ are defined via linear recurrence equations.
\end{definition}

\begin{remark}
A consequence of recognizing a language $\mathcal{L}$ via linear recurrence equations with coefficients in $\mathcal{P}(\Sigma)$ is that the language is automatically partitioned into sets $\mathcal{L}_n$, where $\mathcal{L}_n$ contains all the words of length $n$ in $\mathcal{L}$, i.e., $\mathcal{L}_n$ is the $n$th \textbf{cross-section} of the language.
\end{remark}

Since the operations of the semiring of languages are union and concatenation, it is not difficult to see that we can also define $\mathcal{L}_{\mathcal{A}}$ by using a grammar. We associate a nonterminal symbol $A_i$ to every state $f_i$, except that we denote $A_1$ by $S$, the start symbol, since $f_1$ is the initial state. The set of terminal symbols is $\Sigma$. Consider an arbitrary transition $\displaystyle f_i \stackrel{L_{ij}}{\longrightarrow} f_j$ and suppose that $L_{ij} = \{ a_{ij,1}, a_{ij,2}, \ldots, a_{ij,m} \}$ is non-empty. Then to this transition we associate the productions $A_i \rightarrow a_{ij,1} A_j, \ A_i \rightarrow a_{ij,2} A_j, \ldots , \ A_i \rightarrow a_{ij,m} A_j$. Finally, for every state $f_i \in F_{\mathcal{A}}$ (so $f_i(0) = \{ \varepsilon \}$) we include a production $A_i \rightarrow \varepsilon$. Denote this grammar by $G_{\mathcal{A}}$. Then we will also define $\mathcal{L}_{\mathcal{A}}$ by $\mathcal{L}_{\mathcal{A}} = L(G_{\mathcal{A}})$.

We now present the closure properties of these languages.

\begin{theorem} \label{TheoremClosure}
The languages recognized by counting automata over $\mathcal{P}(\Sigma)$ are closed under union, concatenation, and the Kleene star.
\end{theorem}

\begin{proof}
Assume that $\mathcal{L}_{\mathcal{A}}$ and $\mathcal{L}_{\mathcal{B}}$ are the languages recognized by $\mathcal{A}$ and $\mathcal{B}$, respectively, where $\mathcal{A}$ and $\mathcal{B}$ are counting automata over $\mathcal{P}(\Sigma)$. We will show that we can construct counting automata over $\mathcal{P}(\Sigma)$ that recognize $\mathcal{L}_{\mathcal{A}} \cup \mathcal{L}_{\mathcal{B}}$, $\mathcal{L}_{\mathcal{A}} \cdot \mathcal{L}_{\mathcal{B}}$, and $\mathcal{L}_{\mathcal{A}}^*$.

Let $Q_{\mathcal{A}} = \{ f_1, f_2, \ldots, f_k \}$ be the set of states of $\mathcal{A}$. Assume that $I_{\mathcal{A}} = \{ f_1 \}$ and that $F_{\mathcal{A}}$ is the (nonempty) set of final states of $\mathcal{A}$. Then we know that $\mathcal{L}_{\mathcal{A}}$ is recognized by an automaton like the one in Figure \ref{fig:WeightedLanguages1.eps}, with weights $L_{ij}^{\mathcal{A}}$, for $1 \leq i,j \leq k$. Similarly, let $Q_{\mathcal{B}} = \{ g_1, g_2, \ldots, g_m \}$ be the set of states of $\mathcal{B}$. Suppose that $I_{\mathcal{B}} = \{ g_1 \}$ and that $F_{\mathcal{B}}$ is the (nonempty) set of final states of $\mathcal{B}$. Then $\mathcal{L}_{\mathcal{B}}$ is also recognized by an automaton like the one in Figure \ref{fig:WeightedLanguages1.eps}, but with weights $L_{ij}^{\mathcal{B}}$, for $1 \leq i,j \leq m$.

We first construct an automaton $\mathcal{C}$ recognizing $\mathcal{L}_{\mathcal{A}} \cup \mathcal{L}_{\mathcal{B}}$. Let $Q_{\mathcal{C}} = \{ h \} \cup Q_{\mathcal{A}} \cup Q_{\mathcal{B}}$ and $I_{\mathcal{C}} = \{ h \}$. We let $F_{\mathcal{C}} = \{h\} \cup F_{\mathcal{A}} \cup F_{\mathcal{B}}$ if $f_1 \in F_{\mathcal{A}}$ or $g_1 \in F_{\mathcal{B}}$. Otherwise, $F_{\mathcal{C}} = F_{\mathcal{A}} \cup F_{\mathcal{B}}$. Now we just need to specify the transitions. The new automaton $\mathcal{C}$ will contain all the transitions in $\mathcal{A}$ and in $\mathcal{B}$, plus some new ones. For every transition  $\displaystyle f_1 \stackrel{L_{1j}^{\mathcal{A}}}{\longrightarrow} f_j$, $1 \leq j \leq k$, we add a new transition $\displaystyle h \stackrel{L_{1j}^{\mathcal{A}}}{\longrightarrow} f_j$. Similarly, for every transition  $\displaystyle g_1 \stackrel{L_{1j}^{\mathcal{B}}}{\longrightarrow} g_j$, $1 \leq j \leq m$, we add a new transition $\displaystyle h \stackrel{L_{1j}^{\mathcal{B}}}{\longrightarrow} g_j$. Then $\mathcal{C}$ recognizes $\mathcal{L}_{\mathcal{A}} \cup \mathcal{L}_{\mathcal{B}}$. That is, $\mathcal{L}_{\mathcal{C}} = \mathcal{L}_{\mathcal{A}} \cup \mathcal{L}_{\mathcal{B}}$.

We now construct an automaton $\mathcal{C}$ that recognizes $\mathcal{L}_{\mathcal{A}} \cdot \mathcal{L}_{\mathcal{B}}$. Let $Q_{\mathcal{C}} = Q_{\mathcal{A}} \cup Q_{\mathcal{B}}$ and $I_{\mathcal{C}} = \{ f_1 \}$. We let $F_{\mathcal{C}} = \{ f_1 \} \cup F_{\mathcal{B}}$ if $f_1 \in F_{\mathcal{A}}$ and $g_1 \in F_{\mathcal{B}}$. Otherwise, $F_{\mathcal{C}} = F_{\mathcal{B}}$. As for the transitions in $\mathcal{C}$, we will include all the transitions in $\mathcal{A}$ and in $\mathcal{B}$, as well as some other ones. Assume that $F_{\mathcal{A}} = \{ f_{j_1}, f_{j_2}, \ldots, f_{j_l} \}$. Then, for every state $f_i$, $1 \leq i \leq k$, given the transitions $\displaystyle f_i \stackrel{L_{ij_1}^{\mathcal{A}}}{\longrightarrow} f_{j_1}, \ f_i \stackrel{L_{ij_2}^{\mathcal{A}}}{\longrightarrow} f_{j_2}, \ldots , \ f_i \stackrel{L_{ij_l}^{\mathcal{A}}}{\longrightarrow} f_{j_l}$, we add a new transition $\displaystyle f_i \stackrel{L_i^{\mathcal{A}}}{\longrightarrow} g_1$, where $\displaystyle L_i^{\mathcal{A}} = L_{ij_1}^{\mathcal{A}} \cup L_{ij_2}^{\mathcal{A}} \cup \ldots \cup L_{ij_l}^{\mathcal{A}}$. Finally, if $f_1 \in F_{\mathcal{A}}$, then for every transition $\displaystyle g_1 \stackrel{L_{1j}^{\mathcal{B}}}{\longrightarrow} g_j$, $1 \leq j \leq m$, we add a transition $\displaystyle f_1 \stackrel{L_{1j}^{\mathcal{B}}}{\longrightarrow} g_j$. By construction, we conclude that $\mathcal{L}_{\mathcal{C}} = \mathcal{L}_{\mathcal{A}} \cdot \mathcal{L}_{\mathcal{B}}$.

We conclude the proof by constructing an automaton $\mathcal{C}$ that recognizes $\mathcal{L}_{\mathcal{A}}^*$. Let $Q_{\mathcal{C}} = \{ h \} \cup Q_{\mathcal{A}}$, $I_{\mathcal{C}} = \{ h \}$, and $F_{\mathcal{C}} = \{ h \} \cup F_{\mathcal{A}}$. The automaton $\mathcal{C}$ will contain all the transitions in $\mathcal{A}$, plus some new ones. For each state $f_j$, $1 \leq j \leq k$, given the transition $\displaystyle f_1 \stackrel{L_{1j}^{\mathcal{A}}}{\longrightarrow} f_j$, we add a transition $\displaystyle h \stackrel{L_{1j}^{\mathcal{A}}}{\longrightarrow} f_j$, and for every $f_i \in F_{\mathcal{A}}$, we also add transitions $\displaystyle f_i \stackrel{L_{1j}^{\mathcal{A}}}{\longrightarrow} f_j$ (to the state $f_j$). Then $\mathcal{L}_{\mathcal{C}} = \mathcal{L}_{\mathcal{A}}^*$.
\end{proof}

By combining Theorems \ref{TheoremAutRec} and \ref{TheoremClosure}, we obtain the following.

\begin{corollary} \label{CorollaryClosure}
The languages recognized by systems of linear recurrence equations with coefficients in $\mathcal{P}(\Sigma)$ are closed under union, concatenation and the Kleene star.
\end{corollary}

We are now ready to prove one of our main results.

\begin{theorem} \label{TheoremEquivalance}
A language recognized by a system of linear recurrence equations with coefficients in $\mathcal{P}(\Sigma)$ is regular. Conversely, every regular language is recognized by a system of linear recurrence equations with coefficients in $\mathcal{P}(\Sigma)$.
\end{theorem}

\begin{proof} Suppose that $\mathcal{L}$ is a language recognized by a system of linear recurrence equations with coefficients in $\mathcal{P}(\Sigma)$. Then we know that there is a counting automaton $\mathcal{A}$ over $\mathcal{P}(\Sigma)$ such that $\mathcal{L} = \mathcal{L}_{\mathcal{A}}$. By definition, $\mathcal{L}_{\mathcal{A}}$ is generated by the grammar $G_{\mathcal{A}}$ provided before Theorem \ref{TheoremClosure}. Note that this grammar is regular. Therefore, $\mathcal{L}$ is a regular language.

Now we need to show that if a language is regular, then it is recognized by a system of linear recurrence equations with coefficients in $\mathcal{P}(\Sigma)$. Recall that the set of \textbf{regular languages} over an alphabet $\Sigma = \{ a_1, a_2, \ldots, a_m \}$ is defined by ({\it i}) $\emptyset, \{ \varepsilon \}, \{ a_1 \}, \{ a_2 \}, \ldots, \{ a_m \}$ are regular, and ({\it ii}) if $L_1, L_2$ are regular, then $L_1 \cup L_2, \ L_1 \cdot L_2, \text{ and } L_1^*$ are regular. It is not difficult to find systems of linear recurrence equations with coefficients in $\mathcal{P}(\Sigma)$ that recognize the languages in ({\it i}). Note that
\begin{equation} f_1(n+1) = \emptyset \cdot f_1(n), \quad f_1(0) = \emptyset \end{equation}
recognizes $\emptyset$,
\begin{equation} f_1(n+1) = \emptyset \cdot f_1(n), \quad f_1(0) = \{ \varepsilon \} \end{equation}
recognizes $\{ \varepsilon \}$, and
\begin{equation} \begin{tabular}{l r r l} \hspace{-0.35cm} $f_1(n+1) = \emptyset \cdot f_1(n)$ & \hspace{-0.35cm} $\cup$ \hspace{-0.35cm} & $\{ a_i \} \cdot f_2(n),$ & $f_1(0) = \emptyset$ \\ \hspace{-0.35cm} $f_2(n+1) = \emptyset \cdot f_1(n)$ & \hspace{-0.35cm} $\cup$ \hspace{-0.35cm} & $\emptyset \cdot f_2(n),$ & $f_2(0) = \{ \varepsilon \}$ \end{tabular} \end{equation}
recognizes $\{ a_i \}$, for each $a_i \in \Sigma$. Finally, suppose that $L_1$ and $L_2$ are two languages recognized by systems of linear recurrence equations. By Corollary \ref{CorollaryClosure}, there are systems of linear recurrence equations recognizing $L_1 \cup L_2, \ L_1 \cdot L_2, \text{ and } L_1^*$.
\end{proof}

Theorem \ref{TheoremEquivalance} shows that linear recurrence equations with coefficients in $\mathcal{P}(\Sigma)$ recognize, precisely, the regular languages. Hence, counting automata over $\mathcal{P}(\Sigma)$ recognize the regular languages as well. By Kleene's Theorem, regular languages are recognized by finite automata. Thus, it is natural to translate concepts from finite automata theory to counting automata over $\mathcal{P}(\Sigma)$. For example, we can define what it means for a counting automaton over $\mathcal{P}(\Sigma)$ to be deterministic.


\begin{definition}
Let $\mathcal{A}$ be a counting automaton over $\mathcal{P}(\Sigma)$ and let $Q_{\mathcal{A}} = \{f_1, f_2, \\ \ldots, f_k\}$ be its set of states. We say that $\mathcal{A}$ is \textbf{deterministic} if for every state $f_i$, the transition weight languages $L_{i1}, L_{i2}, \ldots, L_{ik}$ are pairwise disjoint.
\end{definition}

Note that this is equivalent to saying that given $f_i \in Q_{\mathcal{A}}$ and $a \in \Sigma$, $a$ belongs to at most one of the transition weight languages $L_{i1}, L_{i2}, \ldots, L_{ik}$. Hence, our definition coincides with the classical definition of a deterministic automaton (see \cite{Eil74}). Notice we have not discussed how to turn a (nondeterministic) counting automaton into a deterministic one. It should be clear, however, that the techniques to accomplish this from finite automata theory can be applied to counting automata over $\mathcal{P}(\Sigma)$.




\begin{example} \label{ExRecReg} Recurrence Equations and Regular Languages

Let $\mathcal{L} = (ab^*a)^*$ and notice that $\mathcal{L}$ is a regular language. Assume that the alphabet is $\Sigma = \{ a, b \}$. Then $\mathcal{L}$ is recognized by the automaton $\mathcal{A}$ below.
\begin{figure}[H]
\centering
\includegraphics[scale=1]{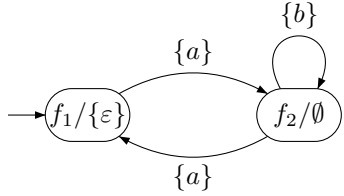}
\caption{Counting automaton $\mathcal{A}$ recognizing $(ab^*a)^*$}
\label{fig:labsars2.eps}
\end{figure}
\noindent Equivalently, $\mathcal{L} = (ab^*a)^*$ is recognized by the system below.
\begin{equation*} \label{eq:SysEx2}
\begin{tabular}{l r l}
$f_1(n+1) =$ & $\{ a \} \cdot f_2(n),$ & $f_1(0) = \{ \varepsilon \}$ \\
$f_2(n+1) =$ & \hspace{-0.45cm} $\{ a \} \cdot f_1(n) \cup \{ b \} \cdot f_2(n),$ & $f_2(0) = \emptyset$
\end{tabular}
\end{equation*}

\noindent We can write $\mathcal{L}$ as $\mathcal{L} = \displaystyle \bigcup_n \mathcal{L}_n$, where $\mathcal{L}_n = f_1(n)$ is the $n$th cross-section of $\mathcal{L}$. If we write the system above in matrix form, as $\begin{bmatrix} f_1(n+1) \\ f_2(n+1) \end{bmatrix} = \begin{bmatrix} \emptyset & \{a\} \\ \{a\} & \{b\} \end{bmatrix} \begin{bmatrix} f_1(n) \\ f_2(n) \end{bmatrix}$, then $\begin{bmatrix} f_1(n) \\ f_2(n) \end{bmatrix} = \begin{bmatrix} \emptyset & \{a\} \\ \{a\} & \{b\} \end{bmatrix}^n \begin{bmatrix} \{\varepsilon\} \\ \emptyset \end{bmatrix}$. Notice that $f_1(1) = \emptyset$, which agrees with the fact that the language $(ab^*a)^*$ has no words of length 1. Similarly, $f_1(4) = \{ aaaa, abba \}$, and note that these are precisely the words of length 4 in $(ab^*a)^*$.

\end{example}

\section{Path-Counting and Self-Counting Automata: Calculating the Density Function of a Regular Language}

In the previous section we saw how we can use weighted automata and linear recurrence equations to recognize a regular language. In particular, we saw how this type of language recognition induces a partition of the language into its cross-sections. And thus, for each $n$, we have a way of generating all the words of length $n$ in the language. In this section we will see that it is possible to output not only the words, but also the number of words, of length $n$ in the language. That is, we show how to calculate the density function of the language. Recall that a word of length $n$ is recognized by a successful path with the same length. With this connection in mind, we start by constructing an automaton that can count, for any given $n$, the number of successful paths of length $n$. In order to do this, we introduce the notion of a path-counting automaton.

\begin{definition}
Given a counting automaton $\mathcal{A}$, its \textbf{path-counting automaton} $\bar{\mathcal{A}}$ is a counting automaton over $\mathbb{N}$ that is able to count the number of successful paths of any given length in $\mathcal{A}$.
\end{definition}

We now show how to construct a path-counting automaton $\bar{\mathcal{A}}$. Assume that the set of states of $\mathcal{A}$ is $\{ f_1, f_2, \ldots, f_k \}$. Then we denote the set of states of $\bar{\mathcal{A}}$ by $\{ \bar{f}_1, \bar{f}_2, \ldots, \bar{f}_k \}$. We let $\bar{f}_i$ be an initial state in $\bar{\mathcal{A}}$ if $f_i$ is an initial state in $\mathcal{A}$. Suppose that $\bar{f}_i(0) = \bar{c}_i$. We let $\bar{c}_i = 1$ if $f_i$ is final. Otherwise, if $f_i$ is non-final, we let $\bar{c}_i = 0$. Finally, consider a transition $\bar{f}_i \stackrel{\bar{a}_{ij}}{\longrightarrow} \bar{f}_j$ in $\bar{\mathcal{A}}$, corresponding to a transition $f_i \stackrel{a_{ij}}{\longrightarrow} f_j$ in $\mathcal{A}$. We let $\bar{a}_{ij} = 1$ if $a_{ij} \neq 0$. Otherwise, if $a_{ij} = 0$, we let $\bar{a}_{ij} = 0$. By Theorem \ref{TheoremAutRec}, the automaton $\bar{\mathcal{A}}$ generates a system
\begin{equation} \label{eq:SysPC} \bar{f}_i(n+1) = \displaystyle \sum_{j=1}^k \bar{a}_{ij} \bar{f}_j(n), \quad \bar{f}_i(0) = \bar{c}_i, \quad 1 \leq i \leq k. \end{equation}
From the way we defined $\bar{c}_i$ and $\bar{a}_{ij}$, a simple proof by induction shows that, if $f_i$ is an initial state, $\bar{f}_i(n)$ equals the number of successful paths of length $n$ that start at $f_i$.

We now define the self-counting automaton.

\begin{definition}
Given a counting automaton $\mathcal{A}$ and its corresponding path-counting automaton $\bar{\mathcal{A}}$, the \textbf{self-counting automaton} $\left(\mathcal{A},\bar{\mathcal{A}}\right)$ is a counting automaton over $K \times \mathbb{N}$ capable of counting its own successful paths of any given length.
\end{definition}

Essentially, $\left(\mathcal{A},\bar{\mathcal{A}}\right)$ is an extension of the automaton $\mathcal{A}$. If the set of states of $\mathcal{A}$ is $\{ f_1, f_2, \ldots, f_k \}$, then the set of states of $\left(\mathcal{A},\bar{\mathcal{A}}\right)$ is $\{ \left(f_1,\bar{f}_1\right), \left(f_2,\bar{f}_2\right), \ldots, \left(f_k,\bar{f}_k\right) \}$. If $f_i$ is an initial state, we let $\left(f_i, \bar{f}_i\right)$ be an initial state, and for every final state $f_j$ of $\mathcal{A}$, we let $\left(f_j, \bar{f}_j\right)$ be a final state of $\left(\mathcal{A},\bar{\mathcal{A}}\right)$. Finally, notice that the weights are also ordered pairs. It is clear that if the weight of $f_i \longrightarrow f_j$ is $a_{ij}$, then the weight of $\left(f_i, \bar{f}_i\right) \longrightarrow \left(f_j, \bar{f}_j\right)$ is $(a_{ij},\bar{a}_{ij})$. We can think of $\left(\mathcal{A},\bar{\mathcal{A}}\right)$ as an extension of $\mathcal{A}$, where the first coordinate keeps track of the \textit{weights of the paths} traversed (thus mimicking $\mathcal{A}$), while the second coordinate keeps track of the \textit{number of paths} traversed.

\begin{remark}
It is not difficult to see that path-counting and self-counting automata can be used to count the number of successful paths of a weighted automaton over \textit{any} alphabet, not just a one-letter alphabet. Since a successful path does not depend on the alphabet used, simply identify all the letters in the alphabet, say $x_1, x_2, \ldots, x_m$, with a letter $x$, and use counting automata.
\end{remark}

We now return to our discussion of formal languages. Recall that a word $w$ of length $n$ is recognized by a counting automaton $\mathcal{A}$ if there a successful path of length $n$ in $\mathcal{A}$ with weight $\{ w \}$. Hence, given a counting automaton $\mathcal{A}$, we would expect the number of words of length $n$ in $\mathcal{L}_{\mathcal{A}}$ to be equal to the number of successful paths of length $n$ in $\mathcal{A}$. That is, we would expect the density function to be $\bar{f}_1(n)$. However, these two quantities could fail to be equal. Notice that (\textit{i}) a path could recognize more than one word, and (\textit{ii}) a word could be recognized by more than one path. The next theorem shows how to correctly define the function that counts the number of words of length $n$ (the density of the language) and the conditions needed on the automaton.

\begin{theorem} \label{TheoremCounting}
Let $\mathcal{A}$ be a deterministic counting automaton over $\mathcal{P}(\Sigma)$. Then the density function of $\mathcal{L}_{\mathcal{A}}$ (the language recognized by $\mathcal{A}$) can be defined via linear recurrence equations with coefficients in $\mathbb{N}$.
\end{theorem}

\begin{proof}
Recall that the function $\bar{f}_1(n)$ in Eq. \ref{eq:SysPC} counts the number of successful paths of length $n$ in $\mathcal{A}$. As we pointed out, this quantity need not be equal to the number of words of length $n$ in $\mathcal{L}_{\mathcal{A}}$. First, we need to account for case (\textit{i}) above, since a path could recognize more than one word. This is because a transition weight may contain more than one letter from $\Sigma$. The recurrence equations that define the density function will be just like the ones in Eq. \ref{eq:SysPC}, except that each coefficient needs to count the number of letters in the corresponding transition weight (instead of just being 0s or 1s). Given a regular language $\mathcal{L}$ recognized by a system
\begin{equation} \label{eq:SysLanguagesCOPY} f_i(n+1) = \displaystyle \bigcup_{j=1}^k L_{ij} \cdot f_j(n), \quad f_i(0) = c_i, \quad 1 \leq i \leq k, \end{equation}
we define the following system of linear recurrence equations
\begin{equation} \label{eq:SysLD} \mathfrak{f}_i(n+1) = \displaystyle \sum_{j=1}^k |L_{ij}| \ \mathfrak{f}_j(n), \quad \mathfrak{f}_i(0) = |c_i|, \quad 1 \leq i \leq k, \end{equation}
where $|S|$ denotes the cardinality of a set $S$.

It is clear that $\mathfrak{f}_1(n)$ is greater than or equal to the number of words of length $n$ in $\mathcal{L}_{\mathcal{A}}$. Note that $\mathfrak{f}_1(n)$ is strictly greater if there is a word recognized by more than one successful path. We claim that, since $\mathcal{A}$ is deterministic, no word of length $n$ is recognized by more than one successful path with the same length, and thus $\mathfrak{f}_1(n)$ gives precisely the number of words of length $n$ in $\mathcal{L}_{\mathcal{A}}$. (Hence, determinism will take care of case (\textit{ii}) above, where a word could be recognized by more than one path.)

Suppose, on the contrary, that there is a word $w$ recognized by more than one successful path. Since $\mathcal{A}$ has only one initial state, then there is at least one state $f_i$ that both paths share, with the property that the transition weights leaving $f_i$ are not pairwise disjoint. This contradicts the fact that $\mathcal{A}$ is deterministic. Thus, we conclude that the density function of $\mathcal{L}_{\mathcal{A}}$ is $\mathfrak{f}_1(n)$.
\end{proof}

\begin{remark}
The \textbf{density function} of a language $\mathcal{L}$ is usually denoted in the literature by $\rho_{\mathcal{L}}(n)$ (see \cite{Roz97}). Formal languages can be classified according to their density. For example, we say that a language has a \textit{constant}, \textit{polynomial}, or \textit{exponential} density if $\rho_{\mathcal{L}}(n) = \mathfrak{f}_1(n)$ has constant, polynomial, or exponential order, respectively.
\end{remark}

\begin{example} Path-Counting Automata, and a Language of Polynomial Density

Consider the regular language $\mathcal{L} = a^*ba^*ba^*$. It easy to see that $\mathcal{L}$ is recognized by the counting automaton $\mathcal{A}$ shown below. (Note that $\mathcal{A}$ is deterministic.)

\begin{figure}[H]
\centering
\includegraphics{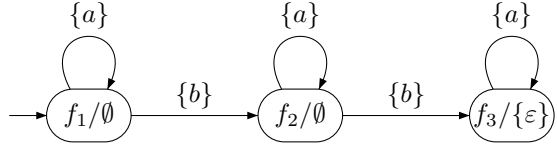}
\caption{Counting automaton $\mathcal{A}$ recognizing $a^*ba^*ba^*$}
\label{fig:asbasbas.eps}
\end{figure}

\noindent It is not difficult to see that the system that defines the density function $\mathfrak{f}_1(n)$ is the one below.

\begin{equation*} \label{eq:SysEx3}
\begin{tabular}{l l l l l}
$\mathfrak{f}_1(n+1) =$ & \hspace{-0.35cm} $\mathfrak{f}_1(n) \ +$ & \hspace{-0.4cm} $\mathfrak{f}_2(n),$ & & $\mathfrak{f}_1(0) = 0$ \\
$\mathfrak{f}_2(n+1) =$ & & \hspace{-0.4cm} $\mathfrak{f}_2(n) \ +$ & \hspace{-0.425cm} $\mathfrak{f}_3(n),$ & $\mathfrak{f}_2(0) = 0$ \\
$\mathfrak{f}_3(n+1) =$ & & & \hspace{-0.425cm} $\mathfrak{f}_3(n),$ & $\mathfrak{f}_3(0) = 1$
\end{tabular}
\end{equation*}

\noindent We could write this system in matrix form, as $\begin{bmatrix} \mathfrak{f}_1(n+1) \\ \mathfrak{f}_2(n+1) \\ \mathfrak{f}_3(n+1) \end{bmatrix} = \begin{bmatrix} 1 & 1 & 0 \\ 0 & 1 & 1 \\ 0 & 0 & 1 \end{bmatrix} \begin{bmatrix} \mathfrak{f}_1(n) \\ \mathfrak{f}_2(n) \\ \mathfrak{f}_3(n) \end{bmatrix}$. Then $\begin{bmatrix} \mathfrak{f}_1(n) \\ \mathfrak{f}_2(n) \\ \mathfrak{f}_3(n) \end{bmatrix} = \begin{bmatrix} 1 & 1 & 0 \\ 0 & 1 & 1 \\ 0 & 0 & 1 \end{bmatrix}^n \begin{bmatrix} 0 \\ 0 \\ 1 \end{bmatrix}$. Alternatively, we could obtain an explicit formula for $\mathfrak{f}_1(n)$ in the following way. Notice that $\mathfrak{f}_3(n) = 1$ for $n \geq 0$, and hence $\mathfrak{f}_2(n) = n$ for $n \geq 0$. Using this, we obtain that $\mathfrak{f}_1(0) = 0$ and, for $n \geq 1$, $\mathfrak{f}_1(n) = \mathfrak{f}_1(n-1) + (n-1)$. Thus, if $n \geq 1$, $\mathfrak{f}_1(n) = \displaystyle \frac{n(n-1)}{2}$. We conclude that the language $a^*ba^*ba^*$ contains exactly $\displaystyle \frac{n(n-1)}{2}$ words of length $n$, for $n \geq 1$.

\end{example}

\begin{example} Path-Counting Automata, and a Language of Exponential Density

Consider again the language $\mathcal{L} = (ab^*a)^*$ from Example \ref{ExRecReg}, recognized by the counting automaton $\mathcal{A}$ in Figure \ref{fig:labsars2.eps}. This automaton is, clearly, deterministic. It is not difficult to see that the density function is defined by the following system.

\begin{equation*} \label{eq:SysEx4}
\begin{tabular}{l r l}
$\mathfrak{f}_1(n+1) =$ & $\mathfrak{f}_2(n),$ & $\mathfrak{f}_1(0) = 1$ \\
$\mathfrak{f}_2(n+1) =$ & \hspace{-0.4cm} $\mathfrak{f}_1(n) + \mathfrak{f}_2(n),$ & $\mathfrak{f}_2(0) = 0$
\end{tabular}
\end{equation*}

\noindent Notice that $\mathfrak{f}_1(n+2) = \mathfrak{f}_2(n+1) = \mathfrak{f}_1(n) + \mathfrak{f}_2(n) = \mathfrak{f}_1(n) + \mathfrak{f}_1(n+1)$, with $\mathfrak{f}_1(0) = 1$ and $\mathfrak{f}_1(1) = 0$. Hence, if $F_n$ denotes the $n$th Fibonacci number, we have that $\mathfrak{f}_1(n) = F_{n-1}$, for $n \geq 1$. We conclude that if $n \geq 1$, the number of words of length $n$ in $(ab^*a)^*$ is $\displaystyle \frac{\varphi^{n-1} - (1-\varphi)^{n-1}}{\sqrt{5}}$, where $\varphi = \displaystyle \frac{1 + \sqrt{5}}{2}$ is the golden ratio.

\end{example}

\section*{Acknowledgements}
This work was made possible, in part, by the PSC-CUNY Grant 60116-39 40.

\bibliographystyle{plain}
\bibliography{TheBibliography}

\end{document}